\documentclass[10pt,draftclsnofoot,journal,letterpaper,onecolumn]{IEEEtran}
\usepackage{amsfonts,amsmath,amssymb,amsthm}

\usepackage[usenames,dvipsnames]{color}
\usepackage{color}

\usepackage[boxed,lined]{algorithm2e}
\usepackage{paralist}

\usepackage{epsfig}
\usepackage{wrapfig}
\usepackage{graphicx}
\usepackage{subcaption}

\usepackage{flushend}
\usepackage{hyperref}
\usepackage{url}

\usepackage{newclude}

\usepackage[noadjust]{cite} 
\usepackage{xfrac}

\newtheorem{definition}{Definition}
\newtheorem{theorem}{Theorem}
\newtheorem{lemma}{Lemma}

\allowdisplaybreaks

\begin{document}
\title{Finite Length Analysis of Caching-Aided Coded Multicasting \footnote{A shorter version of this manuscript appeared in the 52nd Annual Allerton Conference on Communication, Control, and Computing (Allerton), 2014 as an invited paper \cite{Shanmugamallerton}.} }
\author{
\IEEEauthorblockN{Karthikeyan Shanmugam} \\
\IEEEauthorblockA{ Department of Electrical and Computer Engineering \\
University of Texas at Austin, USA \\
\texttt{\{karthish\}@utexas.edu}}\\
\IEEEauthorblockN{Mingyue Ji} \\
\IEEEauthorblockA{ Department of Electrical Engineering \\
University of Souther California, USA\\
\texttt{\{mingyuej\}@usc.edu}}\\
\IEEEauthorblockN{Antonia M.Tulino and Jaime Llorca} \\
\IEEEauthorblockA{ Alcatel Lucent- Bell Labs \\
Holmdel, USA\\
\texttt{\{a.tulino,jaime.llorca\}@alcatel-lucent.com}}\\
\IEEEauthorblockN{Alexandros G. Dimakis} \\
\IEEEauthorblockA{ Department of Electrical and Computer Engineering \\
University of Texas at Austin, USA\\
\texttt{\{dimakis\}@austin.utexas.edu}}
}

\date{\today}

\maketitle

\begin{abstract}
     In this work, we study a noiseless broadcast link serving $K$ users whose requests arise from a library of $N$ files. Every user is equipped with a cache of size $M$ files each. It has been shown that by splitting all the files into packets and placing individual packets in a random independent manner across all the caches, it requires at most $N/M$ file transmissions for any set of demands from the library. The achievable delivery scheme involves linearly combining packets of different files following a greedy clique cover solution to the underlying index coding problem. This remarkable multiplicative gain of random placement and coded delivery has been established in the asymptotic regime when the number of packets per file $F$ scales to infinity.
     
             In this work, we initiate the finite-length analysis of random caching schemes when the number of packets $F$ is a function of the system parameters $M,N,K$. Specifically, we show that existing random placement and clique cover delivery schemes that achieve optimality in the asymptotic regime can have at most a multiplicative gain of $2$ if the number of packets is  sub-exponential. Further, for any clique cover based coded delivery and a large class of random caching schemes, that includes the existing ones, we show that the number of packets required to get a multiplicative gain of $\frac{4}{3}g$ is at least $O((N/M)^g)$. We exhibit a random placement and an efficient clique cover based coded delivery scheme that approximately achieves this lower bound. We also provide tight concentration results that show that the average (over the random caching involved) number of transmissions concentrates very well requiring only polynomial number of packets in the rest of the parameters.
\end{abstract}


\section{Introduction}
 Wireless data traffic is increasing at an alarming pace dominated by video on demand services \cite{index2011global}, and the LTE bandwidth available has not increased to cope with the increasing demand. Recently, to tackle this problem, caching at the wireless edge has been proposed \cite{femto1,molisch2014caching}. Caching could take place at small cell/WiFi access points or at end user devices \cite{golrezaei2012base,ji2013wireless}, by prefetching popular content at off-peak periods to alleviate peak traffic later. It has been shown that in the presence of some form of communication between caches (e.g., device-to-device or D2D communications), caching gains proportional to the aggregate cache size can be obtained \cite{ji2013optimal,ji2014fundamental}. However, when there is no direct communication between caches, conventional caching schemes are limited to local cache hit gains.

 Consider that a set of demands for distinct files arriving at the base station (in a wireless macro cell setting) where each demand corresponds to some user mobile device in the cell. For simplicity, consider the case when user mobile stations are equipped with cache memory. Files stored in (say) user A's cache may or may not be relevant to that user's demand. However, it is possible that another user's demand is stored in the cache. This could benefit the number of files (or its equivalent in terms of bits) that the base station needs to broadcast to satisfy all demands if the cache content of every user is taken into account.  The abstract problem called \textit{index coding} tries to model the aspect of wireless caching systems that do not have local cache hits but their cache (or what is called \textit{side information}) overlaps with other users' demands. In an index coding problem, we have $K$ caching mobile devices served by a noiseless broadcast channel. Each caching device requests a distinct file that is not there in its cache. The broadcast transmissions can be a linear combination of files. Each user recovers its demand using the broadcast transmissions using its cache content. The objective is to find the minimum broadcast transmissions (termed as broadcast rate) given a set of demands and given cache content for each user. This problem is known to be extremely hard to compute and approximate even when linear combinations are done over the binary alphabet \cite{langberg2008hardness}. The problem can be represented as a side information graph where vertices are users and a directed edge means that one user's caching device has the some other users' desired packet as cached information. This problem has received a lot of attention in the information theory literature \cite{effros2012equivalence,bar2011index,maleki2012index} because it  $1$) encapsulates the difficulty of all network coding problems and $2$) any linear coding scheme for index coding is equivalent to a linear interference alignment scheme. We provide an example : User $1$ requests packet $1$ and User $2$ requests packet $2$ and each user has the other users' packet. Although there is no \textit{local} cache hit, the side information present at both users can be used to reduce the number of transmission by $1$ by transmitting the XOR of both packets.
 
   In another line of work, motivated by this 'index coding property' that allows using usersÕ side information to create coded multicast transmissions for users requesting different files, the problem of designing the side information is also considered. This problem is referred to as either the coded caching or the caching-aided coded multicasting problem. Hereafter, we refer to this simply as the caching problem. The setting is same as the index coding problem where there is a library of $N$ files from which user requests arise and every device has a memory of size $M$. The difference is that there is a \textit{placement phase}, which is free of cost, that involves populating all user caches with files from the library. The problem has been studied where order optimal peak broadcast rate for worst-case demand, order optimal average rate for uniform demand distribution and order optimal average rate for Zipf demand distribution for the demand have been characterized. However, all the achievable schemes work in the asymptotic regime when the number of packets per file scales to infinity. In this paper, we consider the case of peak rate over worst-case demand pattern. We show that existing algorithms for placement and delivery give very little gain even when the file sizes are exponentially large in the number of users, and derive lower bounds for a general class of random uncoordinated placement schemes and clique cover based delivery schemes. We also modify existing algorithms to approximately match these bounds. A detailed review of the caching problem is given below. 

\subsection{Related Work}

   In the caching problem, there is a common broadcasting agent serving $K$ users through a noiseless broadcast channel. Every user requests a file from a set of $N$ files. Each file consists of $F$ bits or packets. Every user has a cache of size $M$ files. Files or parts of it ('packets') are placed in every cache \textit{prior} to transmissions assuming that the library of file requests is known in advance. The objective is to design a placement scheme and delivery scheme that optimizes (or approximately optimizes) the maximum number of file transmissions required over all possible demand patterns. This problem has been well studied in the asymptotic regime when $F \rightarrow \infty$.
   
   A deterministic caching and delivery scheme which requires ${K \choose KM/N}$ packets per file to achieve a gain of $KM/N$ was proposed in \cite{maddah2013fundamental}. Following this, a random placement scheme that allows populating user caches independently of each other was proposed in \cite{maddah2013decentralized}. In this uncoordinated placement phase, every user caches $MF/N$ packets of each file $n \in [1:N]$ chosen uniformly at random and independently of other caches. The delivery scheme is a greedy clique cover on the side information graph induced by the underlying index coding problem (refer Section \ref{sec:Definitions}), where a set of packets of possibly different files are XORed if for all packets, at least one user desiring the file corresponding to the packet can recover the desired packet only by using its cache contents. For example if $A+B+C$ was sent, a user wanting $A$ could recover $A$ if the user has $B$ and $C$ stored in its user cache. The \textit{peak broadcast rate} (number of file transmissions) of this scheme was shown to be (in the limit $F \rightarrow \infty$): 
      \begin{equation}\label{eqn:peakexp}
          R_{p} (M) =  \frac{K \left(1- M/N \right)}{\left( KM/N \right)} \left(1- \left( 1-M/N \right)^K \right)
      \end{equation}     
  Here, $R_p(M)$ denotes the peak broadcast rate.  The peak rate means the worst case rate over all demand patterns of the $K$ users from the library. More precisely, this is the average peak rate because it is averaged over the randomness in caching. Note that, if coded multicasting is not used then the rate is given by $K \left(1-M/N \right)$ from the gain due to just local cache hits. It was shown through cut-set bounds that the result in (\ref{eqn:peakexp}) was optimal up to a constant factor.   
 The placement and delivery algorithms that achieve this average peak rate are given in Algorithms \ref{alg:OldPlacement} and \ref{alg:OldDelivery} respectively.
  
\begin{algorithm}
      \KwIn{Parameters $K,M,N$ and $F$.}
        \For {every user $k \in [1:K]$}
          { \For {every file $n \in [1:N]$}
             { Choose a random $MF/N$ subset of $F$ packets of file $n$ and place it in cache $k$.
             }      
          }
     \KwOut{ Cache configuration for every user $k \in [1:K]$.}     
   \caption{OldPlacement (Placement Algorithm in \cite{maddah2013decentralized})}
   \label{alg:OldPlacement}
   \end{algorithm}  
 \begin{algorithm}
      \KwIn{Parameters $K,M,N$ and $F$, caches for all users $k \in [1:K]$ and demand set $\mathbf{d}=\left[d_1,d_2 \ldots d_K \right]$.}
        \For {every subset ${\cal S} \subseteq [1:K]$}
           { Let $V_{k,{\cal S}-k}$ be the vector of packets from file requested by user $k$ but stored exactly in the set of caches ${\cal S}-k$.\\
             Transmit $\oplus_{k \in {\cal S}} V_{k,{\cal S}-k}$. 
           }     
   \caption{OldDelivery (Delivery Algorithm in \cite{maddah2013decentralized}). XORing ($\oplus$'ing) vectors of different lengths means that all shorter vectors are zero padded to match the longest and then XORed.}
   \label{alg:OldDelivery}
   \end{algorithm}

  This was followed by the works of \cite{niesen2013coded} and \cite{ji2014order} where they analyze the case of average number of transmissions when the user demand follows a popularity distribution over the library. Specifically, authors in \cite{ji2014order} consider the case in which file requests follow a Zipf popularity distribution. They provide caching and delivery schemes that achieve order optimal average number of transmissions in the asymptotic regime. The caching distribution, unlike in the worst-case, has to be designed with respect to the collective demand distribution. Interestingly, they also showed that for Zipf parameter between $0$ (uniform popularity) and $1$, even the peak rate scheme given above is sufficient for order optimality in the asymptotic regime $F \rightarrow \infty$. 
  
  \subsection{Our Contribution}
    We consider the caching problem with $K$ users, $N$ files in the library and a cache size of $M$ files. We are interested in the peak broadcast rate (number of file transmissions) for the worst-case demand. Our contributions are:
    \begin{enumerate}
       \item We first show that the existing random uncoordinated placement scheme (Algorithm \ref{alg:OldPlacement}) for this problem and its delivery scheme (Algorithm \ref{alg:OldDelivery}) has a rate above $ \frac{K \left(1-M/N \right)}{2}$ when $F \leq \frac{(N/M)}{K} \exp\left(KM/N \right)$. When compared to the asymptotic result, for a large asymptotic gain when $KN/M$ is $\Omega ((\log K)^2)$, the file size requirement is super-polynomial.
       \item We propose a slightly modified placement scheme (Algorithm \ref{alg:NewPlacement}). We show that the old delivery algorithm (Algorithm \ref{alg:OldDelivery}) coupled with the new placement scheme has similar file size requirements suggesting a needed change in the delivery scheme.
       \item We show that, under any random placement scheme which is independent and symmetric across users (every file packet placement in a user cache is independent of its placement in other caches, every file packet has equal marginal probability of being placed in a cache), any clique cover based scheme (using clique cover on the side-information graph) requires a file size of approximately $O( \frac{g}{K} (N/M)^{g-1})$ for achieving a peak average rate of $\frac{K}{\frac{4}{3}g} \left(1-M/N\right)$. Here, the average is over the random caching involved.
       \item  Since the studied placement schemes are random, it is important to consider the spread in performance due to this randomness. We show that the file size requirements for any clique cover scheme over both random placement schemes (old and the new) is polynomial for the average number of transmissions to concentrate for any demand pattern. It is sufficient to have a file size of $(O (K^3 \log K))$ for the random rate (over the randomness in caching) to be within a constant multiplicative factor from the mean.
       \item We finally exhibit a modified delivery scheme that improves on Algorithm \ref{alg:OldDelivery} through an extra pre-processing step. This modified delivery scheme applied with a specific user grouping along with the new placement scheme provably achieves a rate of roughly $\frac{4K}{3(g+1)}$ with a file size of $O( (\lceil N/M\rceil)^{g+1} (\log (N/M))^{g+2} (2e)^g )$ approximately matching the lower bound. The new placement scheme plays an important role in simplifying  the analysis of this algorithm.
           \end{enumerate}
           
   In Section \ref{sec:Definitions}, we provide the definitions of two random placement schemes (`old' placement scheme used in the literature and a `new' placement scheme) and a delivery scheme previously used in literature. In Section \ref{sec:conc}, for any clique cover scheme, we show that the file size requirements are only polynomial in $K$ for the normalized transmissions in both random placement schemes to concentrate well. In Section \ref{sec:filereq}, we show that the previous delivery scheme, that works asymptotically very well, gives only a constant gain (of $2$) even for exponentially large file sizes. We also show that any clique cover scheme with a random placement scheme that is `symmetric' requires exponential file size in the `target gain'.  For constant target gains, the file size requirement is polynomial in the ratio of library size to the cache memory size per user. In Section \ref{sec:effachiev}, to bridge the gap, we design an efficient clique cover scheme, which together with the new placement scheme, achieves the file size lower bound approximately orderwise.
           
    \section{Definitions and Algorithms} \label{sec:Definitions}
   
     We consider the problem of designing placement and delivery schemes when $K$ users request files from a library of $N$ files ($N>K$) and each user has a cache of size $M$ files. In the placement phase, a file is divided into $F$ packets/bits. Then each packet is placed in different user caches (randomly or deterministically). We are interested in an efficient placement scheme and an efficient delivery scheme consisting of coded XOR transmissions of various packets that optimizes the peak rate over worst-case demands. An efficient delivery scheme computes the coded transmissions needed in time polynomial in parameters $N,K,F,M$, while a placement scheme being efficient ensures $F$ is as small as possible.  Let us denote a set of demands by $\mathbf{d}=\left[ d_1,d_2 \ldots d_K \right],~d_k \in [1:N]$. A packet $f$ belonging to file $n \in [1:N]$ is denoted by $(n,f)$.   
    \begin{definition}
       After a placement scheme, cache configuration ${\cal C}$ is given by the family of sets ${\cal S}_{n,f}$ for all files $n$ and $1 \leq f \leq F$ where $S_{n,f} \subseteq [1:K]$ is the set of user caches in which packet (bit) $f$ of file $n$ is stored.
    \end{definition} 
      Every demand $\mathbf{d}$ and a cache configuration induces a directed side information graph $G=(V,E)$ where there are $KF$ nodes where $(d_k,f)$ is the label for each node representing the $f$-th packet of file $d_k$. There is a directed edge from $(d_k,f)$ to $(d_j,f')$ if the file packet $f'$ of file $d_j$ is stored in the user cache $k$.
      
    \begin{definition}
       A clique cover delivery scheme corresponds to covering nodes of $G$ by cliques. A clique is a set of vertices where there are edges in either direction between all vertices.
    \end{definition} 
   It is easy to see that, XORing all the packets in the clique formed by $(d_{k_1},f_1),(d_{k_2},f_2) \ldots (d_{k_m},f_m)$ implies that user $k_j, 1 \leq j \leq m$ will be able to decode the packet $(d_{k_j},f_j)$ by using all other packets in the XOR from its cache. Note that, here we do not require the demands to be distinct.  

    Let $R^{A}({\cal C},\mathbf{d})$ be the number of normalized transmissions (total number of bits broadcast divided by file size $F$) achieved by a \textbf{given} generic clique cover scheme $A$ on the side information graph induced by the placement ${\cal C}$ and demand $\mathbf{d}$. In the literature, sometimes $R({\cal C},\mathbf{d})$ is also called \textit{broadcast rate} or simply \textit{rate}. We replace $A$ by a short italicized string to denote various algorithms.
            
   \subsection{New Placement and Delivery Schemes}
       
       We first provide our new placement and delivery schemes in Algorithms \ref{alg:NewPlacement} and \ref{alg:NewDelivery} that forms the basis of all our results. The new placement scheme simplifies analysis and helps us to get concentration results. The new delivery scheme is just an efficient polynomial time (in all the parameters) implementation of the old delivery scheme.  
     \begin{algorithm}
      \KwIn{Parameters $K,M,N$ and $F$.}
      Let $F=\lceil N/M\rceil F'$ packets and $F'$ is an integer. Let every file be divided into $F'$ groups each of size $\lceil N/M\rceil$ each.\\
        \For {every user $k \in [1:K]$}
          { \For {every file $n \in [1:N]$}
             { \For {$f' \in [1:F']$}
                { $f'$-th packet of file $n$ in user $k$'s cache is randomly uniformly chosen from the set of $\lceil\frac{N}{M}\rceil$ packets of group $f'$ of file $n$.
                }
             }      
          }
     \KwOut{ Cache configuration for every user $k \in [1:K]$.}     
   \caption{NewPlacement }
   \label{alg:NewPlacement}
   \end{algorithm}  
   
 \begin{algorithm}
      \KwIn{Parameters $K,M,N$ and $F$, caches for all users $k \in [1:K]$ and demand set $\mathbf{d}=\left[d_1,d_2 \ldots d_K \right]$.}
       Let $C= \emptyset$. Let $S_{d_k,f} \subseteq [1:K],~\forall k \in [1:K],~f \in [1:F]$ be the exact subset of users in which the $f$-th packet of file requested by user $k$ is stored. \\
       Let $D \subset [1:K] \times [1:F]$ be the file packets that are stored in the user requesting the corresponding file, i.e. $D=\{(d_k,f): k \in S_{d_k,f} \}$.\\ 
        \For {$(d_k,f) \in [1:K] \times [1:F] - D$}
        {
                   \eIf{$(d_k,f) \notin C$}
                     { 
                          Let $A= \emptyset$. \\
                          \For {$j \in [1:K]-{m} $}
                           {   
                              \If{$\exists (j,f') \notin C~\mathrm{for~some~}f': S_{d_j,f'}=S_{d_k,f} \bigcup {k}-{j}$}
                                  { 
                                     $A \leftarrow A \bigcup (j,f')$ \\
                                   }
                           }
                       Transmit $x_{d_k,f} \oplus_{(j,f') \in A} x_{d_j,f'}$.\\      
                       $C \leftarrow C \bigcup (d_k,f) \bigcup A$.
                     }
                     {  Proceed with the next iteration.\\
                     }    
        }
               
   \caption{NewDelivery}
   \label{alg:NewDelivery}
   \end{algorithm}
   
\textbf{Remark:} $x_{d_k,f}$ in Algorithm \ref{alg:NewDelivery} refers to the content of packet $f$ of the file $d_k$. Also, it is easy to see that Algorithm \ref{alg:NewDelivery} runs in time polynomial in $K$ and $F$.

Let $R^{nd} \left( {\cal C},\mathbf{d} \right)$ denote the normalized transmissions achieved by Algorithm \ref{alg:NewDelivery}. Here, the string $nd$ denotes the delivery scheme in Algorithm \ref{alg:NewDelivery}. Here, $A$ in $ R^{A}({\cal C},\mathbf{d})$ is replaced by a string $nd$ to denote Algorithm \ref{alg:NewDelivery}. Let $\mathbf{d}_u$ denote a set of distinct demand requests by users, i.e. every user requests a distinct file. Let $R^{\mathrm{opt}} \left({\cal C},\mathbf{d} \right)$ denote the number of normalized transmissions under the optimal clique cover scheme on the side information graph due to the cache configuration ${\cal C}$ and the demand pattern $\mathbf{d}$. 

When $\cal{C}$ is chosen randomly, $R^{nd}({\cal C},\mathbf{d})$ is a random variable. Let $\mathbb{E}_{c}$ denote expectation taken over the cache configuration according to a specified random placement described by the string $c$. Further, let $\mathbb{E}_d$ denote expectation over a demand distribution described by $d$. Let $\mathbb{E}_{c,d}$ denote the expectation with respect to both. Let $c_{op}$ denote the `old' random placement according to Algorithm \ref{alg:OldPlacement}. Let $c_{np}$ denote `new' random placement according to Algorithm \ref{alg:NewPlacement}. 

We first note that our new delivery scheme performs identically to Algorithm \ref{alg:OldDelivery}. It is an efficient implementation of the old one.
\begin{theorem}\label{equality}
      The number of transmissions of Algorithm \ref{alg:NewDelivery} is identical to the number of transmissions of Algorithm \ref{alg:OldDelivery} for a given placement and a set of demands.
\end{theorem}
\begin{proof}
   It is easy to see from the description that Algorithm \ref{alg:NewDelivery} is an efficient way to implement Algorithm \ref{alg:OldDelivery}.
\end{proof}

Even the new placement scheme is very similar to the old placement scheme except that it reduces lots of unwanted correlations between different packets belonging to the same file.
This helps us simplify analysis.
\section{Concentration results} \label{sec:conc}

\subsection{New Placement Scheme}
Now, we state Theorem \ref{Thm:Concrate1} about concentration of $R({\cal C},\mathbf{d})$ around its mean for the two placement algorithms. Please note that, the concentration results hold for any delivery algorithm that provides a clique cover on the side information graph induced by ${\cal C}$. Therefore, we do not specify the algorithm used and hence we drop $A$ in $R^A({\cal C},\mathbf{d})$. Before that, we provide a standard technical lemma regarding concentration of martingales.

\begin{lemma}\label{Lem:Azumahoeff}
  (Azuma-Hoeffding, McDiarmid) Consider a random variable $Z=f(X_1,X_2 \ldots X_n)$ where $f(\cdot)$ is a real-valued function and $X_1,X_2 \ldots X_n$ are $n$ random variables. Then $\{ \mathbb{E}[Z \lvert X_0,X_1,X_2 \ldots X_i ] \}_{i=0}^{n}$ forms a martingale. Here, $X_0$ is taken to be a constant. Suppose, these random variables satisfy either one of the following:
   \begin{enumerate}   
     \item \begin{equation}
          \lvert \mathbb{E}[Z \lvert X_1,X_2 \ldots X_i ] - \mathbb{E}[Z \lvert X_1,X_2 \ldots X_{i-1}] \rvert \leq c_i 
      \end{equation} 
    \item (Average Lipschitz Condition)  
       \begin{equation}
          \lvert \mathbb{E}[Z \lvert X_1,X_2 \ldots X_i=a ] - \mathbb{E}[Z \lvert X_1,X_2 \ldots X_{i}=a'] \rvert \leq c_i 
      \end{equation}
    \item (McDiarmid's Inequality) Suppose the set of random variables $\{X_i\}$ are independent.
        \begin{equation}
           \lvert f\left(X_1,X_2, \ldots X_i, \ldots X_n\right) - f \left( X_1,X_2, \ldots X'_i, \ldots X_n\ \right)\rvert \leq c_i
        \end{equation}
    \end{enumerate}  
    
   Then, the following concentration result holds:
   \begin{equation}
       \mathrm{Pr} \left( \lvert Z -\mathbb{E}[Z] \rvert > t \right) \leq 2\exp \left( -\frac{2 t^2}{ \sum_i c_i^2}\right)
   \end{equation}
\end{lemma}

 \begin{theorem}\label{Thm:Concrate1}
    For the random placement (denoted by string $c_{np}$) given in Algorithm \ref{alg:NewPlacement}, any demand distribution denoted by string $d$ (including a singleton distribution on a specific demand) and for any clique cover delivery scheme, we have: 
       \begin{equation}
           \mathrm{Pr}_{c_{np},d}\left(  \lvert R({\cal C},\mathbf{d}) - \mathbb{E}_{c_{np},d} \left[ R ({\cal C},\mathbf{d}) \right] \rvert \geq \epsilon \mathbb{E}_{c_{np},d}[R ({\cal C},\mathbf{d})]  \right) \leq 2\exp \left(- \frac {2 \epsilon^2 \left(\mathbb{E}_{c_{np},d}[R ({\cal C},\mathbf{d})] \right)^2 F }{K \frac{N^2}{M^2}}  \right)
       \end{equation}
 \end{theorem}
 \begin{proof} 
   We use martingale analysis on a generic clique cover algorithm. We denote any generic clique cover algorithm by algorithm A. Clearly, the number of transmissions:
      \begin{equation}
      R \left( {\cal C},\mathbf{d} \right) = h \left( S_{d_1,1},S_{d_1,2} \ldots S_{d_1,F} \ldots S_{d_k,j} \ldots S_{d_K,F} \right)
      \end{equation}
       for some function $h(\cdot)$ where $S_{d_k,f} \subseteq [1:K]$ is the subset of users caches in which the file packet $f$ of file $d_k$ is cached. In other words, the number of transmissions given ${\cal C}$ and $\mathbf{d}$ is determined fully by specifying $S_{d_k,f},~\forall d_k \in [1:K], ~f \in [1:F]$. Further, $S_{d_k,f}$ is dependent on both the cache configuration ${\cal C}$ and the demand $\mathbf{d}$. 
    
     Now, we apply Lemma \ref{Lem:Azumahoeff} with random variables $X_{d_k,f}$ set to $S_{d_k,f},~\forall k \in [1:K],~f \in [1:F]$ and $Z$ is set to $R \left({\cal C},\mathbf{d}\right)$. Consider the expression for a specific $(d_k,f)$:
     \begin{equation} \label{eq:ALC}
    c_{d_k,f}= \lvert \mathbb{E}_{c_{np},d} \left[ R \left( {\cal C}, \mathbf{d} \right) \right \lvert S_{1,1}, S_{1,2} \ldots S_{d_k,f}]- \mathbb{E}_{c_{np},d} \left[ R \left( {\cal C}, \mathbf{d} \right) \right \lvert S_{1,1}, S_{1,2} \ldots S_{d_k,f-1}] \rvert
    \end{equation}
       In the first term in (\ref{eq:ALC}), let us assume that the choice of $S_{d_k,f}$ is consistent with the previous choices of $S_{1,1},S_{1,2} $ $\ldots S_{d_k,f-1}$. In Algorithm \ref{alg:NewPlacement}, every file is grouped into $F'$ groups each of size $\lceil\frac{N}{M}\rceil$. Let us assume that packet $f$ of the file $d_k$ belongs to group $g$. The placement of file packets is independent across the groups $g$. The choice of $S_{d_k,f}$ affects the placement of at most $\lceil\frac{N}{M}\rceil$ packets belonging to group $g$ of file $d_k$. Other file packet placements are unaffected. Let $V=\{ (d_k,f): f\mathrm{~belongs~to~group~}g \}$ be the set of bits in the same group $g$ of file $d_k$.
       
         Consider a new Algorithm B: 1) Run clique cover algorithm A excluding the file packets in $V$. 2) Then, transmit the file packets in $V$ separately. The file packets in $V$ is not used in Step $1$ of algorithm B. Let $R_B({\cal C},\mathbf{d})$ be the number of transmissions in Step $1$ of Algorithm B. Clearly, the following holds:
         \begin{equation}
             R_B({\cal C},\mathbf{d}) \leq R({\cal C}, \mathbf{d}) \leq R_B({\cal C},\mathbf{d}) +  \frac{ \lceil\frac{N}{M}\rceil} {F}
         \end{equation}
         
          This is because, the first step of Algorithm B employs the same clique cover scheme as Algorithm A and operates on a sub-graph induced by the file packets in the system other than $V$. Therefore, the number of transmissions has to be reduced. Further, adding the packets of $V$ in step $2$ is a sub-optimal way of improving algorithm $A$.
         
          Therefore, both $\mathbb{E}_{c_{np}} \left[ R \left( {\cal C}, \mathbf{d} \right) \right \lvert S_{1,1}, S_{1,2} \ldots S_{d_k,f}=A]$ and $\mathbb{E}_{c_{np}} \left[ R \left( {\cal C}, \mathbf{d} \right) \right \lvert S_{1,1}, S_{1,2} \ldots S_{d_k,f-1}] $ are at most $\frac{\lceil\frac{N}{M} \rceil}{F}$ away from the performance of Step $1$ of algorithm $B$ averaged over their respective cache realizations. Further, the performance of Step $1$ of algorithm $B$ ($R_B({\cal C},\mathbf{d})$) is independent of the choice of $S_{d_k,f}$ because the possibly affected file packets (in set $V$) have been removed in Step $1$ of algorithm $B$. Therefore, $c_{d_k,f} \leq \frac{N}{MF},~\forall k$ in Lemma \ref{Lem:Azumahoeff}.
 
  Hence, applying Lemma \ref{Lem:Azumahoeff}, we have: 
     \begin{align}
            \mathrm{Pr}_{c_{np},d}\left(  \lvert R({\cal C},\mathbf{d}) - \mathbb{E}_{c_{np},d} \left[ R ({\cal C},\mathbf{d}) \right] \rvert \geq \epsilon \mathbb{E}_{c_{np},d}[R ({\cal C},\mathbf{d})]  \right) & \leq 2\exp \left(- \frac {2 \epsilon^2 \left(\mathbb{E}_{c_{np},d}[R ({\cal C},\mathbf{d})] \right)^2 }{ KF \left( \frac{N}{MF} \right)^2}  \right) \nonumber \\
             \hfill       & \leq \exp \left(- \frac {2 \epsilon^2 \left(\mathbb{E}_{c_{np},d}[R ({\cal C},\mathbf{d})] \right)^2 F }{ K \left( \frac{N}{M} \right)^2 } \right)
     \end{align}   
 
 \end{proof}
 
 \textbf{Remark:} The above result shows that when $F \geq \frac{8}{\epsilon^2}\frac{KN^2}{M^2 \left(\mathbb{E}_{c_{np},d}[R ({\cal C},\mathbf{d})] \right)^2 } \log K$, then with probability at least $1-1/K^8$, $R({\cal C},\mathbf{d}) \in \left[ (1-\epsilon) \mathbb{E} \left( R({\cal C},\mathbf{d}) \right), (1+\epsilon)\mathbb{E} \left( R({\cal C},\mathbf{d}) \right) \right]$. But, for these algorithms to have a non-trivial gain even when $F \rightarrow \infty$, $KM/N \geq 1$ (see (\ref{eqn:peakexp})). This means that $N/M \leq K$. Hence, $F = O(K^3 \log K)$ is sufficient for $R({\cal C},\mathbf{d})$ to be below $(1+\epsilon) \mathbb{E}_{c_{np},d} \left( R({\cal C},\mathbf{d}) \right)$ with very high probability.

    \subsection{Old Placement}
         \begin{theorem} \label{thm:concold}
       Under the old placement scheme $c_{op}$ and any demand distribution on $\mathbf{d}$ (including a singleton distribution on a specific demand), the number of transmissions for any clique cover scheme satisfies the following concentration result:
         \begin{equation}
              \mathrm{Pr}_{c_{op},d}\left(  \lvert R({\cal C},\mathbf{d}) - \mathbb{E}_{c_{op},d} \left[ R ({\cal C},\mathbf{d}) \right] \rvert \geq \epsilon \mathbb{E}_{c_{op},d} \left[ R ({\cal C},\mathbf{d}) \right]  \right) \leq 2\exp \left(- \frac {2 \epsilon^2 \left(\mathbb{E}_{c_{op},d} \left[ R ({\cal C},\mathbf{d}) \right] \right)^2 F }{K(K+1)^2}  \right)
         \end{equation}
     \end{theorem}
     \begin{proof}
       We use a martingale argument as before. Consider a generic clique cover scheme implemented by Algorithm $A$. As before, $R({\cal C},d)$ is a function of $\{S_{d_k,f} \}_{k \in [1:K], f \in [1:F]}$ where $S_{d_k,f} \subseteq [1:K]$ is the subset of caches in which $f$-th file packet of file $d_k$ is stored. In this proof, $S_{d_k,f}$ is with respect to the old placement scheme $c_{op}$. Consider the following:
          \begin{equation} \label{eqn:diffold}
               c_{d_k,f}=  \lvert \mathbb{E}_{c_{op},d} \left[ R \left( {\cal C}, \mathbf{d} \right) \right \lvert S_{1,1}, S_{1,2} \ldots S_{d_k,f}=S]- \mathbb{E}_{c_{op}} \left[ R \left( {\cal C}, \mathbf{d} \right) \right \lvert S_{1,1}, S_{1,2} \ldots S_{d_k,f}=S'] \rvert
          \end{equation} 
     
      When placement of file packets $S_{1,1} \ldots S_{d_k,f-1}$ are fixed, let $n_j$ packets be left among $MF/N$ packets allocated for the file requested by user $k$ in user cache $j$. When $S_{d_k,f}=S \subseteq [1:K]$, let the number of packets left for file $d_k$ at user $j$'s cache be $n_j- \mathbf{1}_S(j) $. $\mathbf{1}_S(j)=1$ if $j \in S$ and $0$ otherwise. In any realization, satisfying the first conditioning where $S_{d_k,f}=S$ in (\ref{eqn:diffold}), for user cache $j$, $n_j- \mathbf{1}_S(j)$ packets are randomly chosen from the remaining $F-f$ packets belonging to file requested by user $k$. Similarly, $n_j- \mathbf{1}_S'(j )$ packets are randomly chosen from the remaining $F-f$ packets belonging to file requested by user $k$ for the second conditioning when $S_{d_k,f}=S'$ in (\ref{eqn:diffold}).
      
       Now, consider the following scheme (delivery+placement scheme) with genie aided transmission as follows: 1) Genie provides the packet $(d_k,f)$ to all users during decoding but is not placed in any user's cache. 
    2) We perform the old placement. 3) Since genie provides $(d_k,f)$ to all users for 'free', if $S_{d_k,f}=S$ immediately after step $2$, we delete the file packet $(d_k,f)$ from caches of users in set $S$ and for all users in $S$ replace it with a new random file packet from the remaining file packets of the file $d_k$, different from the ones placed according to the old placement including the file packet $(d_k,f)$ in step $2$.  We use $c_{gp}$ to denote 'genie-aided placement' as described above.     
 
    Let us contrast this with the old placement: In the old placement when $S_{d_k,f}=S$, for every $j \in S$, $(d_k,f)$ is placed in cache $j$. In the genie aided case, this space at user $j$ has been taken over by a randomly chosen packet from file $d_k$ that has not been used by the old placement at user $j$. Everything else remains the identical to old placement. Genie helps every user get the file packet $(d_k,f)$ saving one packet transmission. In addition to the old placement, $\lvert S \rvert$ additional random file bits are stored. This could at most save $\lvert S \rvert \leq K$ packet transmissions from that of the old placement under any clique covering scheme. It is because every user gets at most one extra packet in its cache from file $d_k$ that it originally had.  Therefore,
        \begin{equation}\label{Eqn:diff}
            0 \leq \mathbb{E}_{c_{op},d} \left[ R \left( {\cal C}, \mathbf{d} \right) \lvert S_{1,1}, \ldots S_{d_k,f-1}, S_{d_k,f}=S \right] - \mathbb{E}_{c_{gp},d} \left[ R \left( {\cal C}^{g}, \mathbf{d} \right) \lvert S_{1,1}, \ldots S_{d_k,f-1}, S_{d_k,f}= S \right]  \leq \frac{K+1}{F}.
        \end{equation}
       Note that, in the above equation, conditioning of $c_{gp}$ till $S_{d_k,f}$ is with respect to step $2$ ( just immediately after the old placement) in the genie aided-placement. 
        
        Now, we view the genie aided-placement using a second view: It is exactly identical to performing the old placement except that, only for file $d_k$, for each user a random set of $MF/N$ packets are drawn from $F-1$ packets that excludes the packet $(d_k,f)$. The genie aided placement ignores conditioning $S_{d_k,f}$ of the old placement. This means that:
         \begin{equation}\label{eqn:same}
            \mathbb{E}_{c_{gp},d} \left[ R \left( {\cal C}, \mathbf{d} \right) \lvert S_{1,1}, \ldots S_{d_k,f-1}, S_{d_k,f}= S\right]= \mathbb{E}_{c_{gp},d} \left[ R \left( {\cal C}, \mathbf{d} \right) \lvert S_{1,1}, \ldots S_{d_k,f-1}, S_{d_k,f}= \emptyset \right]
         \end{equation}
        
       Clearly, from (\ref{eqn:same}), $\mathbb{E}_{c_{gp},d} \left[ R \left( {\cal C}, \mathbf{d} \right) \lvert S_{1,1}, \ldots S_{d_k,f-1}, S_{d_k,f}= S\right]$ is independent of $S$. Therefore, 
          \begin{align}
            c_{d_k,f} &=  \lvert \mathbb{E}_{c_{op},d} \left[ R \left( {\cal C}, \mathbf{d} \right) \right \lvert S_{1,1}, S_{1,2} \ldots S_{d_k,f}=S]- \mathbb{E}_{c_{op},d} \left[ R \left( {\cal C}, \mathbf{d} \right) \right \lvert S_{1,1}, S_{1,2} \ldots S_{d_k,f}=S'] \rvert \nonumber \\
               \hfill & = \lvert \left( \mathbb{E}_{c_{op},d} \left[ R \left( {\cal C}, \mathbf{d} \right) \right \lvert S_{1,1}, S_{1,2} \ldots S_{d_k,f}=S]-\mathbb{E}_{c_{gp},d} \left[ R \left( {\cal C}, \mathbf{d} \right) \lvert S_{1,1}, \ldots S_{d_k,f-1}, S_{d_k,f}= \emptyset \right] \right) - \nonumber \\
                   \hfill & \left( \mathbb{E}_{c_{op},d} \left[ R \left( {\cal C}, \mathbf{d} \right) \right \lvert S_{1,1}, S_{1,2} \ldots S_{i,f}=S'] - \mathbb{E}_{c_{gp},d} \left[ R \left( {\cal C}, \mathbf{d} \right) \lvert S_{1,1}, \ldots S_{i,f-1}, S_{i,f}= \emptyset \right]   \right)\rvert \nonumber \\
               \hfill & \overset{a}   \leq \frac{K+1}{F}    
          \end{align}   
   Justification: (a): This is because of (\ref{Eqn:diff}) and (\ref{eqn:same}). Applying the Average Lipschitz condition in Lemma \ref{Lem:Azumahoeff} with $c_{d_k,f}$, we have the result claimed.
   \end{proof}
   \textbf{Remark:} As before, when $F = \Omega \left(\frac{1}{\epsilon^2} K^3 \log K \right)$, the rate $R({\cal C},d)$ under old placement is within $(1+\epsilon)$ about the expected value multiplicatively.
   
 \section{File Size Requirements under new and Old Placements }\label{sec:filereq}
 
 \subsection{Requirements for Algorithm \ref{alg:NewDelivery} under New Placement} \label{sec:ndnewp}
  Given any cache configuration ${\cal C}$ and demand $\mathbf{d}$, according to Theorem \ref{equality}, the number of transmissions of Algorithm \ref{alg:OldDelivery} and Algorithm \ref{alg:NewDelivery} have identical number of transmissions. Therefore, the expected number of transmissions for Algorithm \ref{alg:OldDelivery} under the new placement algorithm is $\mathbb{E}_{c_{np}} \left[R^{nd} \left( {\cal C},\mathbf{d} \right) \right]$.
  
  Consider any demand distribution for $\mathbf{d}$. Let $\mathbf{1}_{d_k,f,g}^{{\cal S}}$ be the indicator that the packet $f$ ($1 \leq f \leq \lceil\frac{N}{M}\rceil$) of group $g$ in file $d_k$ is placed exactly in the set ${\cal S} \subseteq [1:K]$ of caches. Now, we have :
   \begin{align}
      \mathbb{E} \left[ \mathbf{1}_{d_k,f,g}^{{\cal S}} \right] = \left( \frac{1}{\lceil\frac{N}{M}\rceil} \right)^{\lvert {\cal S} \rvert} \left(1- \frac{1}{\lceil N/M\rceil} \right)^{K - \lvert {\cal S} \rvert}
   \end{align}
   Also, for $\cal S$ that contain $k$, we have:
   \begin{align}
     \lvert V_{k, {\cal S}-k} \rvert = \sum \limits_{g=1}^{F'} \sum \limits_{f=1}^{\lceil N/M\rceil} \mathbf{1}_{d_k,f,g}^{{\cal S}-k}
   \end{align}
   The variables in different groups are independent. $\sum \limits_{f=1}^{\lceil N/M\rceil} \mathbf{1}_{d_k,f,g}^{{\cal S}-k} $ is a bernoulli variable since from every group at most one file bit is stored in a given cache.
   \begin{equation}\label{eqn:mu}
    \mathbb{E}\left[\sum \limits_{f=1}^{\lceil N/M\rceil} \mathbf{1}_{d_k,f,g}^{{\cal S}-k} \right] = \lceil N/M\rceil  \left( \frac{1}{\lceil\frac{N}{M}\rceil} \right)^{\lvert {\cal S}\rvert-1} \left(1- \frac{1}{\lceil N/M\rceil} \right)^{K - \lvert {\cal S} \rvert+1} = \mu (\lvert {\cal S} \rvert). 
    \end{equation} 
    Therefore,  $\lvert V_{k, {\cal S}-k} \rvert $ is a binomial random variable with $F'$ trials and probability of success $\mu$.
   
   The expected number of transmissions for all stages ${\cal S}$ in Algorithm \ref{alg:OldDelivery} with respect to the new placement is given by:
      \begin{equation}\label{eqn:exp}
      \mathbb{E}_{\mathrm{c}_{np}}\left[R^{nd}({\cal C},\mathbf{d})\right]= \sum \limits_{{\cal S} \neq \emptyset}  \frac{\mathbb{E}[\max \limits_{k \in {\cal S}} \lvert V_{k, {\cal S}-k} \rvert]}{F'\lceil\frac{N}{M}\rceil} 
      \end{equation}
   In the limit when $F' \rightarrow \infty$, the binomial random variables are concentrated at their mean and therefore, 
      \begin{equation}
        \lim \limits_{F \rightarrow \infty} \mathbb{E}\left[R^{nd}({\cal C},\mathbf{d})\right] \approx \sum \limits_{{\cal S} \neq \emptyset} \frac{\mu (\lvert {\cal S} \rvert)}{F'\lceil\frac{N}{M}\rceil} = \frac{K \left(1- \frac{1}{\lceil\frac{N}{M}\rceil} \right)}{\left( K \frac{1}{\lceil\frac{N}{M}\rceil}\right)} \left(1- \left( 1- \frac{1}{\lceil\frac{N}{M}\rceil} \right)^K \right) \approx R_p(M)
      \end{equation}
      
%
 
 Now, let us consider the case of distinct demands by all users, denoted by $\mathbf{d}_u$. We assume that $N > K$ here. We are interested in the question: How far are $\mathbb{E}_{c_{np}}\left[R^{nd}({\cal C},\mathbf{d}_u)\right]$ and the limiting peak rate $R_p\left( M \right)$ for finite $F$?

   
    Let $\mathrm{Bi} \left( n,p \right)$  be the binomial distribution for $n$ trials and probability of success $p$. When the demands of all users are different, $\lvert V_{k,{\cal S}-k} \rvert$ is distributed according to $\mathrm{Bi}(F', \mu (\lvert {\cal S} \rvert))$ and for different $k$, the binomial random variables are independent. Now, we show that coding gain is roughly at most $2$, even when $F'$ is exponential in the targeted gain $t=\frac{K}{\lceil \frac{N}{M}\rceil}$
    
    \begin{theorem}\label{thm:newpllow}
     Let $N>K$. Then, $ \mathbb{E}_{c_{np}} \left[ R^{nd} \left( {\cal C},\mathbf{d}_{u} \right) \right]  \geq \frac{1}{2} \left( 1- \frac{M}{N} \right) K$ when $F \leq \frac{\lceil N/M \rceil}{2K} \left(1 - \frac{1}{\lceil N/M \rceil} \right) \exp \left( 2t \left( 1 - \frac{t}{K} \right)\left( 1 - \frac{1}{K} \right) \right)$.
    \end{theorem}
   \begin{proof}
       \begin{align}\label{eqn:lower bound}
           \mathbb{E}_{c_{np}} \left[ R^{nd} \left( {\cal C},\mathbf{d}_{u} \right) \right] &=  \sum \limits_{{\cal S} \neq \emptyset} \frac{\mathbb{E} \left[ \max \limits_{k \in {\cal S}}\lvert V_{k,{\cal S}-k} \rvert \right]}{F'\lceil N/M\rceil}  \nonumber \\
            \hfill                                                                                                       & \overset{a}= \sum \limits_{{\cal S} \neq \emptyset} \mathbb{E} \left[ \frac{ \max \limits_{k \in {\cal S}} Y_k }{F'\lceil N/M\rceil}    \right]~~~~~Y_k \sim\mathrm{Bi}(F',\mu(\lvert {\cal S} \rvert)),~ Y_k = \sum \limits_{f=1}^{F'} Y_{k,f} \nonumber \\
            \hfill                                                                                                       &  \geq \sum \limits_{{\cal S} \neq \emptyset} \frac{ \mathrm{Pr} \left( \bigcup \limits_{k \in {\cal S},f \in [1:F']} Y_{k,f} > 0 \right)} { F'\lceil N/M\rceil  } \nonumber \\
           \hfill                                                                                                  & \overset{b}\geq \sum \limits_{{\cal S} \neq \emptyset} \frac{\left[ \sum \limits_{k \in {\cal S}, f \in [1:F']} \mathrm{Pr} \left( Y_{k,f} > 0 \right) - \sum \limits_{i,j \in {\cal S},f_1,f_2 \in [1:F'],(i,f_1) \neq (i,f_2)} \mathrm{Pr} \left( Y_{i,f_1}>0 \bigcap Y_{j,f_2} >0 \right)  \right]} {F'\lceil N/M\rceil}  \nonumber \\          
                  \hfill                                                                                                 & \overset{c} \geq \frac{\sum \limits_{s=1}^K {K \choose s} \left( sF' \mu(s) - \frac{1}{2}s^2(F')^2 \left( \mu(s) \right)^2 \right) }{F'\lceil N/M\rceil} \nonumber \\
                   \hfill                                                                                                & \geq   \frac{\left(1 - \frac{1}{\lceil\frac{N}{M}\rceil} \right)}{\frac{1}{\lceil\frac{N}{M}\rceil} } \sum \limits_{s=1}^K {K \choose s} s \left( \frac{1}{\lceil\frac{N}{M}\rceil} \right)^s \left(1 - \frac{1}{\lceil\frac{N}{M}\rceil} \right)^{K-s} - \frac{\sum \limits_{s=1}^K {K \choose s}s^2 F' \left( \mu(s) \right)^2 }{2\lceil\frac{N}{M}\rceil} \nonumber \\
                                             \hfill                                                                             & \geq  K \left( 1- \frac{1}{\lceil N/M\rceil}\right) - \frac{1}{2}F' \lceil\frac{N}{M}\rceil \sum \limits_{s=1}^K {K \choose s}s^2 \left(\frac{1}{\lceil\frac{N}{M}\rceil} \right)^{2(s-1)} \left(1 - \frac{1}{\lceil\frac{N}{M}\rceil} \right)^{2(K-s+1)} \nonumber \\
                                              \hfill                                                                    & \geq K \left(1- \frac{1}{\lceil \frac{N}{M}\rceil} \right)-  \frac{1}{2}\left(1 -\frac{1}{\lceil \frac{N}{M} \rceil} \right)^{2K} F' \lceil\frac{N}{M}\rceil \sum \limits_{s=1}^K {K \choose s} s^2   \left( \frac{ \frac{1}{\lceil\frac{N}{M}\rceil} } { 1 - \frac{1}{\lceil\frac{N}{M}\rceil} } \right)^{2(s-1)}   \nonumber \\
                                                \hfill                                                        & \overset{d} \geq  K \left(1- \frac{1}{\lceil \frac{N}{M}\rceil} \right) \nonumber \\
                                                                                                                       & ~ -  \frac{1}{2}F' \lceil\frac{N}{M}\rceil  \left(K \left( 1 - \frac{1}{\lceil\frac{N}{M}\rceil} \right)^2+K(K-1)\left( \frac{1}{\lceil\frac{N}{M}\rceil}\right)^2\right) \left( \left( 1 - \frac{1}{\lceil\frac{N}{M}\rceil} \right)^2+ \left( \frac{1}{\lceil\frac{N}{M}\rceil}\right)^2  \right)^{(K-1)}   \nonumber \\
                                                        \hfill                                                &  \geq  K \left(1- \frac{1}{\lceil \frac{N}{M}\rceil} \right) -  \frac{1}{2}F' K  \left(\frac{K}{t}+t\right) \exp \left( -2t \left(1-\frac{t}{K} \right) \left(1 -\frac{1}{K} \right) \right)     \nonumber \\   
                                                        \hfill                                               & \overset{e} \geq K \left(1- \frac{1}{\lceil \frac{N}{M}\rceil} \right) - F' K^2 \exp \left( -2t \left(1-\frac{t}{K} \right) \left(1 -\frac{1}{K} \right) \right)                                                                                                                                                                                                   
\end{align}
      (a) This is because every $\lvert V_{k,{\cal S}-k}\rvert$ is a sum of $F'$ independent Bernoullli random variables and the set of Bernoulli variables across different values of $k$ are independent because the demands are distinct. (b) We use the following Bonferroni inequality: $\mathrm{Pr} \left( \bigcup \limits_{i=1}^n A_i \right) \geq \sum \limits_{i=1}^n \mathrm{Pr}\left(A_i \right) - \sum \limits_{i \neq j} \mathrm{Pr}\left(A_i \bigcap A_j \right)$. (c) $\mu(s)$ is defined in (\ref{eqn:mu}) and $Y_{i,f_1}$ and $Y_{i,f_2}$ are independent if $(i,f_1) \neq (i,f_2)$.  (d) We use: $\sum \limits_{s \geq 1} {K \choose s} s^2 p^{s-1} = \frac{d}{dp} \left[ p \frac{d}{dp} \left(1+p \right)^k \right]$ and $(1+p)^{K-2} \leq (1+p)^{K-1}$ for $p>0$ and further simplification. (e) We use: $1-x \leq \exp (-x),~\forall x>0$.
    
    This implies that when $F' \leq \frac{1}{2K} \left(1 - \frac{1}{\lceil N/M \rceil} \right) \exp \left( 2t \left( 1 - \frac{t}{K} \right)\left( 1 - \frac{1}{K} \right) \right)$, the expected number of normalized transmissions for Algorithm \ref{alg:OldDelivery} for distinct requests under the new placement scheme given by Algorithm \ref{alg:NewPlacement} is at least $\frac{1}{2} \left( 1- \frac{M}{N} \right) K$. 
        \end{proof}
   Therefore, there is very little coding gain if we do not have exponential number of file packets (exponential in $t$).
     
    \subsection{Requirements for Algorithm \ref{alg:NewDelivery} under Old Placement}   
          Let $R^{nd} \left( {\cal C},\mathbf{d} \right)$ denote the normalized number of transmissions for the new delivery scheme. Let $\mathbf{1}_{d_k,f}^{{\cal S}}$ be the indicator random variable that bit $f$ of file $d_k$ is stored exactly in user caches in the set ${\cal S} \subset [1:K]$. When $k \in {\cal S}$, let us define: 
           \begin{equation}
             \lvert V_{k,{\cal S}-k}\rvert =\sum \limits_{f=1}^F \mathbf{1}_{d_k,f}^{{\cal S}-k}
           \end{equation}
    Here, $\mathbb{E}_{c_{op}}\left[ \mathbf{1}_{d_k,f}^{{\cal S}-k} \right]=  \left( \frac{1}{\frac{N}{M} }\right)^{\lvert {\cal S}\rvert-1} \left(1- \frac{1}{ N/M} \right)^{K - \lvert {\cal S} \rvert+1} = \mu' (\lvert {\cal S} \rvert)$. Consider the case when user demands are distinct (implicitly $N>K$). The following theorem shows that the coding gain is at most $2$ even when the file size is exponential in the targeted gain $t=\frac{K}{N/M}$.
    
    \begin{theorem}
     Let $N>K$. Then, $ \mathbb{E}_{c_{op}} \left[ R^{nd} \left( {\cal C},\mathbf{d}_{u} \right) \right]  \geq \frac{1}{2} \left( 1- \frac{M}{N} \right) K$ when $F \leq \frac{N/M}{2K} \left(1 - \frac{1}{ N/M} \right) \exp \left( 2t \left( 1 - \frac{t}{K} \right)\left( 1 - \frac{1}{K} \right) \right)$.
    \end{theorem}
    \begin{proof}
     We have the following chain:
            \begin{align}\label{eqn:lower bound2}
           \mathbb{E}_{c_{op}} \left[ R^{nd} \left( {\cal C},\mathbf{d}_{u} \right) \right] &=  \sum \limits_{{\cal S} \neq \emptyset} \frac{\mathbb{E} \left[ \max \limits_{k \in {\cal S}} \lvert V_{k,{\cal S}-k} \rvert \right]}{F}  \nonumber \\
            \hfill                                                                                                       & = \sum \limits_{{\cal S} \neq \emptyset} \mathbb{E} \left[ \frac{ \max \limits_{k \in {\cal S}} \sum \limits_{f=1}^F \mathbf{1}_{d_k,f}^{{\cal S}-k}}{F}    \right] ~~\nonumber \\
            \hfill                                                                                                       &  \geq \sum \limits_{{\cal S} \neq \emptyset} \frac{ \mathrm{Pr} \left( \bigcup \limits_{k \in {\cal S},f \in [1:F]}  \mathbf{1}_{d_k,f}^{{\cal S}-k}> 0 \right)} { F} \nonumber \\
           \hfill                                                                                                  & \geq \sum \limits_{{\cal S} \neq \emptyset} \frac{\left[ \sum \limits_{k \in {\cal S}, f \in [1:F]} \mathrm{Pr} \left( \mathbf{1}_{d_k,f}^{{\cal S}-k}> 0 \right) - \sum \limits_{i,j \in {\cal S},f_1,f_2 \in [1:F], (i,f_1) \neq (j,f_2)} \mathrm{Pr} \left( \mathbf{1}_{d_i,f_1}^{{\cal S}-i}>0 \bigcap  \mathbf{1}_{d_j,f_2}^{{\cal S}-j} >0 \right)  \right]} {F}  \nonumber \\          
                  \hfill                                                                                                 & \overset{a} \geq \frac{\sum \limits_{s=1}^K {K \choose s} \left( sF \mu'(s) - \frac{1}{2}s^2(F)^2 \left( \mu'(s) \right)^2 \right) }{F} \nonumber \\
                                                \hfill                                                        & \overset{b} \geq  K \left(1- \frac{M}{N} \right) \nonumber \\
                                                                                                                       & ~ -  \frac{1}{2}F \left(K \left( 1 - \frac{M}{N} \right)^2+K(K-1)\left( \frac{M}{N}\right)^2\right) \left( \left( 1 - \frac{M}{N} \right)^2+ \left( \frac{M}{N}\right)^2  \right)^{(K-1)}   \nonumber \\
                                                        \hfill                                                &  \geq  K \left(1- \frac{1}{\frac{N}{M}} \right) -  \frac{1}{2}F \left(K+t^2\right) \exp \left( -2t \left(1-\frac{t}{K} \right) \left(1 -\frac{1}{K} \right) \right)     \nonumber \\   
                                                        \hfill                                               & \overset{c} \geq K \left(1- \frac{1}{ \frac{N}{M}} \right) - F Kt \exp \left( -2t \left(1-\frac{t}{K} \right) \left(1 -\frac{1}{K} \right) \right)                                                                                                                                                                                                   
\end{align}
    Justifications are : 
    
    (a)- Since the demands are distinct, when $i \neq j$, $d_i \neq d_j$. Therefore the corresponding indicators are independent. Therefore, $ \mathrm{Pr} \left( \mathbf{1}_{d_i,f_1}^{{\cal S}-i}>0 \bigcap  \mathbf{1}_{d_j,f_2}^{{\cal S}-j} >0 \right) =  \mathrm{Pr} \left( \mathbf{1}_{d_i,f_1}^{{\cal S}-i}>0 \right)\mathrm{Pr} \left( \mathbf{1}_{d_j,f_2}^{{\cal S}-j}>0 \right),~i \neq j $. This probability is easily seen to be $\left( \mu'(s) \right)^2$. When $i=j$ and $f_1 \neq f_2$, we have: 
     \begin{align}
        \mathrm{Pr} \left( \mathbf{1}_{d_i,f_1}^{{\cal S}-i}>0 \bigcap  \mathbf{1}_{d_i,f_2}^{{\cal S}-j} >0 \right) & =  \mu'(s) \left(  \frac{M-\frac{1}{F}}{N- \frac{1}{F}}\right)^{s-1}  \left( 1- \frac{M}{N- \frac{1}{F}}\right)^{K-s+1}  \nonumber \\
          \hfill & \overset{d} \leq \left( \mu'(s) \right)^2
     \end{align}
    (d) is because $\frac{M- \frac{1}{F}}{N- \frac{1}{F}} \leq \frac{M}{N},~ M \leq N$. Step (a) (and its justification) in the above chain of inequalities is the main difference between old placement ($c_{op}$) and new placement ($c_{np}$).
    
    (b)- This follows the exact same derivation as in the proof of Theorem \ref{thm:newpllow} except for the factor $\lceil N/M \rceil$.
    
   (c)- This follows from: $K+t^2 \leq 2Kt$.
   
   This implies that when $F \leq \frac{1}{2t} \left(1 - \frac{1}{\lceil N/M \rceil} \right) \exp \left( 2t \left( 1 - \frac{t}{K} \right)\left( 1 - \frac{1}{K} \right) \right)$, the expected number of normalized transmissions for Algorithm \ref{alg:OldDelivery} for distinct requests under the old placement scheme given by Algorithm \ref{alg:OldPlacement} is at least $\frac{1}{2} \left( 1- \frac{M}{N} \right) K$. This implies that there is very little coding gain ($t = K \frac{M}{N}$) even when we have file size exponential in $t$.
    \end{proof}
    \subsection{Requirements for any Clique Cover Delivery Scheme}
    
  Let $c_{up}$ denote a random independent and symmetric placement algorithm that has the following properties:
  \begin{enumerate}
   \item For any packet $(n,f)$, the probability of placing this in a user cache $k$ is independent of placing it in all other caches. 
    \item Placing of packets belonging to different files in the same cache is independent. 
    \item The probability of placing a packet equals $M/N$ for a given cache. 
    \end{enumerate}
    Now, we have the following result on any clique cover scheme on the side information graph induced by random caching algorithm $c_{up}$ and a unique set of demands $\mathbf{d}_u$.  
    \begin{theorem}
       When user demands are distinct, for any clique cover algorithm on the side information graph induced by the random cache configuration due to $c_{up}$, if $ \mathbb{E}_{c_{up}} \left(R \left( {\cal C}, \mathbf{d}_u \right) \right) \leq \frac{K(1-M/N)}{\frac{4}{3}g}$ for any $g>2$, then we need the number of file packets $F \geq  \frac{g}{2et} \left( \frac{N}{M} \right)^{g-2}$ where $t=KM/N$.
 Clearly, these bounds apply to both $c_{op}$ and $c_{np}$.
     \end{theorem}
    \begin{proof}
       We show this by contradiction. Let us assume that $ \mathbb{E}_{c_{up}} \left(R \left( {\cal C}, \mathbf{d}_u \right) \right) \leq \frac{K(1-M/N)}{\frac{4}{3}g}$. This implies: $\mathrm{Pr}_{c_{up}} \left( R \left( {\cal C},\mathbf{d}_u \right) \leq \frac{K (1-M/N)}{g} \right) \geq \frac{1}{4}$ (by Markov's Inequality). The number of transmissions $R \left( {\cal C}, \mathbf{d}_u \right) \leq \frac{K(1-M/N)}{g}$ implies that there is at least there is one clique of size $g$ in the side information graph $G$ induced by ${\cal C}$ and $\mathbf{d}_u$. Given cache configuration ${\cal C}$ and distinct demands $\mathbf{d}_u$, let $n_{g}$ denote the number of distinct cliques of size $g$.   So we have the following chain of inequalities:
        \begin{align}
            \mathrm{Pr}_{c_{up}} \left( R \left( {\cal C},\mathbf{d}_u \right) \leq \frac{K (1-M/N)}{g} \right) & \leq \mathrm{Pr}_{c_{up}} \left( \mathrm{there~is~one~clique~of~size~}g \right) \nonumber \\
             \hfill                                                                                                                                          &  \overset{a} \leq \mathbb{E}_{c_{up}} \left( n_g \right)  \nonumber \\
             \hfill                                                                                                                                          & \overset{b} \leq {K \choose g} F^g \left(\frac{M}{N} \right)^{g(g-1)} \nonumber \\
             \hfill                                                                                                                                           & \overset{c} \leq \left( \frac{Ke}{g} \right)^g F^g \left(\frac{M}{N} \right)^{g(g-1)} \nonumber \\
              \hfill                                                                                                                                          & \leq \left( \frac{Ket}{gt} \right)^g F^g \left(\frac{t}{K} \right)^{g(g-1)} \nonumber \\
              \hfill                                                                                                                                         & \leq    \left(  \frac{etF}{g}\right)^g \left( \frac{t}{K}\right)^{g(g-2)} \label{eqn:probbound}           
        \end{align}
      When $F < \frac{g}{2et} \left( \frac{N}{M} \right)^{g-2} $ and $g>2$, then probability given by (\ref{eqn:probbound}) is strictly less than $1/4$ contradicting the assumption. Therefore, the desired implication follows.  Justifications are: (a) $Pr\left( X \geq 1 \right) \leq \mathbb{E}[X]$. (b) There are ${K \choose g}$ ways of choosing $g$ users caches. Since all demands are distinct, there are $F^g$ ways of choosing $g$ file packets belonging to the files requested by the chosen users. $\left( M/N \right)^{g-1}$ is the probability that a file packet wanted by one of the users is present in $g-1$ other user caches. Since the demands are distinct and placement of packets belonging to different files are different, the probability of forming a $g$-clique is given by $\left( M/N \right)^{g(g-1)}$. (c) ${K \choose g} \leq \left( \frac{Ke}{g} \right)^g$.
    \end{proof}
    \textbf{Note:} We would like to note that $c_{up}$ represents a broad set of schemes where every file packet is placed in a cache independently of its placement elsewhere and no file packet is given undue importance over other packets belonging to the same file.
    
    \section{Efficient Achievable Schemes} \label{sec:effachiev}
    
    \subsection{Deterministic Caching Scheme with User Grouping:}
    
    Now, briefly we would like to explore what can be said about the file size requirements of deterministic placement schemes. In this section, we describe a variation of the deterministic caching scheme in \cite{maddah2013fundamental} that requires a similar file size requirement as the previous section for a target gain of $g$. However, it is not clear if, for a clique cover scheme at the delivery stage, this is the best one can do with deterministic caching schemes.  In other words, a lower bound for deterministic caching scheme similar to the one above is not known.
    
       Now, we give a description of a deterministic caching and delivery scheme that requires $F = {K \choose g } $ packets to get a gain of $g+1$. This follows directly from the deterministic scheme of \cite{maddah2013fundamental}.
   For ease of exposition we describe it here: For every file, split the file into ${K \choose g}$ packets. For every subset $G \subset [1:K]$ such that $\lvert G \rvert=g$, we place the corresponding packet in the user caches in the subset $G$. The total number of files per user cache is $N \frac{{K-1 \choose g-1}}{{K \choose g}} = \frac{g N}{K} \leq M $. This satisfies the memory constraint because the gain $g \leq KM/N$. Following the same arguments in \cite{maddah2013fundamental}, it is easy to show that the peak transmission rate is at most : $\frac{K-g}{g+1}$.  
   
    Now, we show a slight modification of the deterministic caching scheme mentioned above which (approximately order wise) matches the lower bound in the previous section. Let us divide the users into groups of size $K'=g \lceil N/M\rceil$ and then apply the caching and delivery scheme for each group separately. The number of file packets required is $F= {K' \choose g}$. The memory constraint would be satisfied when $g \leq K'M/N=g \lceil N/M\rceil \left(M/N \right)$ which is true. Now, coded multicasting is done within every user group. The total number of transmissions is: $ \frac{K}{K'} \frac{K'-g}{g+1} = \frac{K}{g+1}\left( 1- \frac{1}{\lceil N/M \rceil} \right)$. This requires ${K' \choose g}= O \left( \left( \lceil N/M\rceil e \right)^g \right)$ packets.
   
   \subsection{New Randomized Delivery scheme}
      For the deterministic scheme described previously, similar to the one in \cite{maddah2013fundamental}, it is necessary to refresh (possibly) all the caches in a specific way when users leave or join the system that requires coordination among the caches. Now, we show that under an uncoordinated random caching scheme given by the new placement scheme in Algorithm \ref{alg:NewPlacement} and a new randomized clique cover algorithm, it is possible to have an average peak rate (with respect to all the randomness) of about $\frac{K}{g+1}$ when $F =O\left( g{K \choose g } \log K \right) $.  First, we introduce the new randomized delivery algorithm that we use to prove the above assertion. The new randomized delivery algorithm has a preprocessing step, that we call the `pull-down phase', in addition to Algorithm \ref{alg:NewDelivery}.
      
     \begin{algorithm}
      \KwIn{Parameters $K,M,N,g$ and $F$, caches for all users $k \in [1:K]$ and demand set $\mathbf{d}=\left[d_1,d_2 \ldots d_K \right]$.}
      Let $S_{d_k,f} \subseteq [1:K],~\forall k \in [1:K],~f \in [1:F]$ be the exact subset of users in which the $f$-th packet of file requested by user $k$ is stored. \\
       \For {$(d_k,f) \in [1:K] \times [1:F] $}
        { 
           \If{$ \lvert S_{d_k,f} \rvert \geq g+1 $}
           {  $S_{d_k,f} \leftarrow$ a random $g$-subset of $S_{d_k,f}$           
           }                   
        }               
      Run Algorithm \ref{alg:NewDelivery} with this new cache configuration.  
   \caption{ModifiedDelivery}
   \label{alg:ModifiedDelivery}
   \end{algorithm} 
   
   \textbf{Remark:} Algorithm \ref{alg:ModifiedDelivery} emulates a virtual alteration of the cache configuration. The change in $S_{d_k,f}$ happens in such a way that the algorithm pretends that a file packet is being stored in a subset of a set of caches where it has been actually stored. We use the same notation $S_{d_k,f}$ to represent such a `virtual cache configuration' that will be used for the delivery. For example, if a particular packet was stored in caches $\{1,2,3,4,5,6\}$ and if $g=3$, a random subset from this is chosen. So the resultant virtual cache configuration could be $\{1,2,3\}$ after this virtual re-assignment. The re-assignment phase is what we call the `pull down' phase. This will allow us to `target' the gain $g$ (which is typically a lot lesser compared to $t=KM/N$) more effectively if we use Algorithm \ref{alg:ModifiedDelivery} for delivery.
   
   Let ${\cal R}^{md}(\cal C,\mathbf{d})$ be the random number of transmissions under Algorithm \ref{alg:ModifiedDelivery} given a fixed cache configuration ${\cal C}$ and demand pattern $\mathbf{d}$. In this case, there is further randomness that is a part of the delivery phase. Let $\mathbb{E}^{md}\left( {\cal R}^{md}(\cal C,\mathbf{d}) \right)$ denote the expected number of transmissions with respect to the randomness in Algorithm \ref{alg:ModifiedDelivery}.
      
   We need the following lemma from \cite{raab1998balls} (see proof of Theorem 1).   
  \begin{lemma} \label{ballsbinstheorem}
   \cite{raab1998balls} Consider $m$ balls being thrown randomly uniformly and independently into $n$ bins. When $m= r(n) n \log n $ where $r(n)$ is $O ((log (n))^p)$ for some positive integer $p$, then maximum number of balls in any bin is at most $r(n) \log n (1+ 2 \frac{\sqrt{2}}{r(n)} )$ with probability at least $1- \frac{1}{n^2}$.
  \end{lemma}
      
    \begin{theorem} \label{thm:modalgperf}
          Using the randomized Algorithm \ref{alg:NewPlacement} for the placement scheme and the randomized Algorithm \ref{alg:ModifiedDelivery} for delivery, for any set of demands $\mathbf{d}$, the average peak rate, with respect to all the randomness (randomness in both delivery and placement ) is given by $\mathbb{E}^{md}_{c_{\mathrm{np}}} ({\cal R}^{md}({\cal C},\mathbf{d})) \leq \frac{4}{3}\frac{K}{g+1}(1+o(1))$ and the number of file packets needed is $F = O\left( {K \choose g } (\log ({K \choose g}))^2  \lceil N/M \rceil\right)$ when $ 2 \leq g \leq \frac{K}{3\lceil N/M\rceil}, \lceil N/M \rceil \leq \frac{K}{\frac{27}{4} \log K  }, N >K $.
    \end{theorem}
 \begin{proof}
     According to the placement scheme given by Algorithm \ref{alg:NewPlacement}, every file is made up of $F'$ groups of file packets. Each group has size $\lceil N/M \rceil$. Let us consider the $j$-th  packet of every group. There are $F'$ such file packets. We will first analyze assuming that algorithm \ref{alg:ModifiedDelivery} uses only the $F'$ file packets formed by considering only the $j$-th packet from every group. We will finally add up the number of transmissions for every set of $F'$ packets formed using the differently numbered packet (for all $j \in [1:\lceil N/M \rceil]$) from every group. Clearly, this is suboptimal. Therefore, this upper bounds the performance of Algorithm \ref{alg:ModifiedDelivery}.
     
     Consider a file $n$. Let $G_j^n$ be the set of $F'$ packets, each of which is the $j$-th packet from every group of file $n$ according to the groups formed during placement algorithm \ref{alg:NewPlacement}. Let $S_{n,f,j} \subseteq [1:K]$ be the subset of user caches where the $f$-th packet in $G^n_j$ is stored. Here, $1 \leq f \leq F'$ indicates the position among $F'$ packets formed by taking the $j$th packet from very group. Given a user cache $k$, the placement of packets from the set $G_j^n$ are mutually independent of each other. The marginal probability of placing it is given by $ \frac{1}{\lceil N/M \rceil}$.  The placement is also independent across caches. Therefore, the number of user caches in which a particular packet in $G_j^n$ is placed is a binomial random variable $\mathrm{Bi} \left(K, \frac{1}{\lceil N/M \rceil} \right)$ where $\mathrm{Bi}\left( m,p \right)$ is a binomial distribution with $m$ independent trails each with probability $p$. Therefore, by chernoff bounds (see Pg. $276$ \cite{jukna2001extremal}), $\mathrm{Pr} \left( \lvert S_{n,f,j} \rvert < g \right) \leq \exp \left( - \frac{K}{\lceil N/M \rceil} \left(1 -\frac{g \lceil N/M\rceil}{K} \right)^2  \right) \leq \exp \left( - \frac{4K}{9\lceil N/M \rceil} \right)$. Here, we have used the fact that $g \leq \frac{K}{3 \lceil N/M\rceil}$. Therefore, for any $j$ (by Markov's Inequality), 
     
     \begin{equation} \label{eqn:highlevel}
         \mathrm{Pr} \left( \sum \limits_{f=1}^{F'} \mathbf{1}_{\lvert S_{n,f,j} \rvert < g} > 3F' (g+1)K \lceil N/M \rceil \exp \left(-\frac{4K}{9\lceil N/M\rceil}\right)  \right) \leq \frac{1}{3(g+1)K \lceil N/M \rceil}
     \end{equation}
  $\lceil N/M \rceil \leq \frac{K}{
 \frac{27}{4}\log K }$ and $g \leq \frac{K}{3\lceil N/M\rceil}$ implies the following condition (which can be verified by algebra):
  \begin{equation} \label{condn}
   (g+1)K^2 \lceil N/M\rceil <  \exp \left(\frac{4 K}{ 9 \lceil N/M \rceil }\right). 
   \end{equation}  
     If a file bit is stored in $p$ caches, then the file packet is said to be on level $p$.  This implies, that with high probability, $\left(1- 3(g+1)K \lceil N/M\rceil \exp \left(-\frac{4 K}{9 \lceil N/M\rceil}\right)\right)F'$ file packets belonging to file $n$ from $G_j^n$ is stored at a level above or equal to $g$.  We will first compute the number of transmissions due to applying Algorithm \ref{alg:ModifiedDelivery} only on the file packets in $\{{d_k,f,j}\}_{1 \leq k \leq K, f \in [1:F']}$ for a particular $j$. 
     
     We start by considering a fixed demand pattern $\mathbf{d}= \{d_1,d_2 \ldots d_K \}$. Applying union bound with (\ref{eqn:highlevel}) over at most $K$ files in the demand $\mathbf{d}$, we have:
      \begin{equation} \label{eqn:unionbnd}
         \mathrm{Pr} \left( \exists k 
         \in [1:K]: \sum \limits_{f=1}^{F'} \mathbf{1}_{\lvert S_{d_k,f,j} \rvert < g} > 3(g+1)F' K \lceil N/M \rceil \exp \left(-\frac{4K}{9\lceil N/M\rceil}\right)  \right) \leq \frac{1}{3(g+1)\lceil N/M \rceil}
      \end{equation} 
     Now, consider Algorithm \ref{alg:ModifiedDelivery}. The first few steps of the algorithm, denoted henceforth as `pull down' phase, brings every file packet stored above level $g$ to level $g$. Consider a file packet $(d_k,f,j)$ before the beginning of Algorithm \ref{alg:ModifiedDelivery}. Given that the packet $(d_k,f,j)$ is at a level above $g$, after the `pull down' phase, the probability that it occupies any of the ${K \choose g}$ subsets is equal. This is because prior to the pull down phase, the probability that the file packet being stored in a particular cache is independent and equal to $\frac{1}{\lceil N/M \rceil}$. Consider the $F'$ file packets $\{(d_k,f,j)\},~1 \leq f \leq F'$. Clearly, the probability of any one of them (say $(d_k,f,j)$) occupying a given set of $g$ caches, after the pull down phase, is independent of the occupancy of all other file packets $\{(d_k,f',j)\}_{f' \neq f}$. Let $S^a_{d_k,f,g}$ denote the occupancy after the pull down phase. Therefore after the pull down phase, which is applied only to the files in the demand vector $\mathbf{d}$, 
       \begin{equation}\label{eqn:eqoccup}
          \mathrm{Pr} \left(S_{d_k,f,j}^a = B \lvert ~~\lvert S^a_{d_k,f,j} \rvert >g, \{S^a_{d_k,f',g}\}_{f \neq f'}  \right) = \frac{1}{{K \choose g}} ,~ \forall B \subseteq {[1:K] \choose g}, ~ k \in [1:K],~ 1 \leq j \leq \lceil N/M \rceil 
       \end{equation}
     
   After the pull down phase in Algorithm \ref{alg:ModifiedDelivery},  we compute the number of transmissions of Algorithm \ref{alg:NewDelivery} using the modified $S^a_{d_k,f,j}$ after the pull down phase. It has been observed that Algorithm \ref{alg:NewDelivery} is equivalent to Algorithm \ref{alg:OldDelivery}. After the pull down phase, all the files packets are present at file level $g$ or below. Let us set $F' =  c {K \choose g} \left(\log ({K \choose g}) \right)^2$ for some constant $c>0$.  After the pull down phase, let $V_{k,{\cal S}-k}^{j}$ be the set of file packets in $G^{d_k}_j$ requested by user $k$ but stored exactly in the cache of users specified by ${\cal S}-k$. With respect to only the file packets $\bigcup \limits_{k \in [1:K]} G^{d_k}_j $, the number of transmissions of Algorithm \ref{alg:NewDelivery} is given by:
   
   \begin{align}\label{eqn:notrans}
      \mathrm{No.~of~trans}(j) & = \sum \limits_{{\cal S} \neq \emptyset} \frac{  \max \limits_{k \in {\cal S}} \lvert V_{k, {\cal S}-k}^{j} \rvert}{F'} \nonumber \\
         \hfill                               & \overset{a}= \sum \limits_{{\cal S} \neq \emptyset, \lvert {\cal S} \rvert \leq g+1} \frac{ \max \limits_{k \in {\cal S}} \lvert V_{k, {\cal S}-k}^{j} \rvert}{F'} \nonumber \\
         \hfill                               &  = \sum \limits_{\lvert {\cal S} \rvert = g+1} \frac{ \max \limits_{k \in {\cal S}} \lvert V_{k, {\cal S}-k}^{j} \rvert}{F'} +  \sum \limits_{\lvert {\cal S} \rvert \leq g} \frac{ \max \limits_{k \in {\cal S}} \lvert V_{k, {\cal S}-k}^{j} \rvert}{F'}
   \end{align}
    (a)- This is because after the pull down phase, all the relevant file packets are at a level at most $g$. Consider the event $E$ that $b=\left(1-3(g+1)K \lceil N/M \rceil \exp \left(-\frac{4K}{9\lceil N/M\rceil}\right) \right)F'$ bits of $G^{d_i}_j$ for all $i$ are stored at a level above $g$ before the beginning of Algorithm \ref{alg:ModifiedDelivery}. Conditioned on this event being true, by (\ref{eqn:eqoccup}), the pull down phase is equivalent to throwing $b$ balls independently and uniformly randomly into ${K \choose g}$ bins.  Using (\ref{condn}) and the fact that $F' = c{K \choose g} (\log ({K \choose g}))^2$, the pull down phase is akin to throwing $m= (1 - 3(g+1)K \lceil N/M \rceil \exp \left(-\frac{4K}{9\lceil N/M\rceil}\right)) F' \geq c\left(1-  \frac{3(g+1)}{K} \right) \log n  (n \log n)$ balls into $n = {K \choose g}$  bins. In fact, the $m$ balls of file $d_k$ are being thrown independently and uniformly randomly into bins satisfying ${\cal S}-k: \lvert{\cal S}\rvert =g+1,~ k \in {\cal S}$.  We apply, Lemma \ref{ballsbinstheorem} for a particular user $k$ to obtain: 
    
    \begin{align} \label{eqn:max}
    \mathrm{Pr} \left( \max \limits_{{\cal S}:\lvert {\cal S} \rvert =g+1,~ k \in {\cal S}} \frac{ \lvert V_{k,{\cal S}-k}^j \rvert }{F'} \geq \frac{\frac{m}{n}\left(1+ O \left( \frac{1}{\log K} \right) \right)}{F'} \lvert E \right) \leq \frac{1}{{K \choose g}^2}
    \end{align}
    Please note that $r(n)$ as in Lemma \ref{ballsbinstheorem} is $O(\log K)$. Now, applying a union bound over all users $k$ to (\ref{eqn:max}), we have:
    \begin{align}
     \mathrm{Pr} \left( \exists k \in [1:K]: \max \limits_{{\cal S}:\lvert {\cal S} \rvert =g+1,~ k \in {\cal S}} \frac{ \lvert V_{k,{\cal S}-k}^j \rvert }{F'} \geq \frac{\frac{m}{n}\left(1+ O \left( \frac{1}{\log K} \right) \right)}{F'} \lvert E \right) \leq \frac{K}{{K \choose g}^2}
    \end{align}
    This implies that all $V_{k,{\cal S} - k}$ are bounded in size. Therefore, we have the following:
    
   \begin{align}  \label{eqn:upperbnd1}
        1 - \frac{K}{ {K \choose g}^2} &\leq \mathrm{Pr} \left( \sum \limits_{\lvert {\cal S} \rvert =g+1} \frac{ \max \limits_{k \in {\cal S}} \lvert V_{k,{\cal S}-k}^j \rvert }{F'} \leq {K \choose g+1} \frac{\frac{m}{n}\left(1+ O \left( \frac{1}{\log K} \right) \right)}{F'} \lvert E \right) \nonumber \\
        \hfill &\overset{a}= \mathrm{Pr} \left(\sum \limits_{\lvert {\cal S} \rvert =g+1} \frac{ \max \limits_{k \in {\cal S}}  \lvert V_{k,{\cal S}-k}^j \rvert }{F'} \leq \frac{K-g}{g+1} \left(1+O \left( \frac{1}{\log K} \right)\right) \lvert E \right) . 
   \end{align}
   (a) is because: $1 \geq \frac{m}{F'} \geq 1- \frac{3(g+1)}{K}$ implying $\frac{m}{F'} (1+ \frac{1}{O(\log K)}) = (1+\frac{1}{O(\log K)})$.
    Putting together (\ref{eqn:upperbnd1}), (\ref{eqn:notrans}) and (\ref{eqn:unionbnd}), we have:
      \begin {equation}
        \mathrm{Pr} \left( \mathrm{No.~of~trans}(j) \leq \frac{K-g}{g+1}\left(1+ O \left( \frac{1}{\log K} \right) \right) + 2K^2 \lceil \frac{N}{M} \rceil e^{ -\frac{4K}{9\lceil N/M\rceil}} \right) \geq \left(1-\frac{1}{2(g+1) \lceil N/M \rceil} \right) \left(1 - \frac{1}{ {K \choose g}^2} \right)
      \end{equation}
     Union bounding over all $1 \leq j \leq \lceil N/M \rceil$, we have:
      \begin{equation}
          \mathrm{Pr} \left( \exists j: \mathrm{No.~of~trans}(j) > \frac{K-g}{g+1}\left(1+ O \left( \frac{1}{\log K} \right) \right) + 3(g+1)K^2 \lceil \frac{N}{M} \rceil e^{ -\frac{4K}{9\lceil N/M\rceil}} \right) \leq \frac{1}{3(g+1)}+ \frac{\lceil N/M \rceil}{ {K \choose g}^2} 
      \end{equation}
   From (\ref{condn}), we have $3(g+1)K^2 \lceil \frac{N}{M} \rceil e^{ -\frac{4K}{9\lceil N/M\rceil}} < 3$. Now combining transmissions for different $j$ and normalizing by $\lceil N/M \rceil$, we have:
   \begin{equation} \label{eqn:highprobresult}  
     \mathrm{Pr}^{md}_{c_{\mathrm{np}}} \left( R^{md} \left( \cal C, \mathbf{d} \right) > \frac{K-g}{g+1} (1+o(1)) \right) \leq \frac{1}{3(g+1)} + \frac{\lceil N/M \rceil}{ {K \choose g}^2} =\frac{1}{3(g+1)}+O(1/K)
   \end{equation}   
  
  In the above bad event, the number of  transmissions (normalized) needed is at most $K$. Therefore, we have:
    \begin{equation}
          \mathbb{E}^{md}_{c_{\mathrm{np}}} \left[ R^{md} \left( {\cal C}, \mathbf{d} \right) \right] \leq \frac{K-g}{g+1} (1+o(1)) (1 -  \frac{1}{3(g+1)}-O(1/K)) +  \left(\frac{1}{3(g+1)}+O(1/K) \right) K \leq \frac{4}{3} \frac{K}{g+1} (1+o(1))                                                                                                                                                                
     \end{equation}    
 \end{proof}

  \subsection{Grouping into smaller user groups: approximately achieving the lower bound}    
       We now propose a user grouping scheme similar to the one for the deterministic caching scheme which can achieve the same average number of transmissions as the scheme mentioned in the previous section but with improved file size requirement almost matching the lower bound.
       
       We group users in groups of size $K'= \lceil N/M\rceil 3g (\log (N/M))$ and apply the new placement scheme (Algorithm \ref{alg:NewPlacement}) and delivery scheme of Algorithm \ref{alg:ModifiedDelivery} to each of the user groups. It can be seen that $K'$ satisfies the conditions: $ e \leq \lceil N/M \rceil \leq \frac{K'}{\frac{27}{4} \log K' }$ and $ 7 \leq g \leq \min \{ \frac{K'}{3 \lceil N/M \rceil}, \frac{\left(\frac{N}{M} \right)^2}{3 \log \left( N/M \right)} \}$. Therefore, Theorem \ref{thm:modalgperf} is applicable. For every group, the average number of transmissions for a particular demand configuration is at most $ \frac{4}{3}\frac{K'}{g+1} \left(1+o(1) \right)$. Adding over all groups, we have the following theorem:       
   
       \begin{theorem}
              Let the placement scheme be that of Algorithm \ref{alg:NewPlacement}. For any target gain $7 \leq g \leq \frac{\left(\frac{N}{M} \right)^2}{3 \log \left( N/M \right)}$ and $\lceil N/M \rceil \geq  e  $, let the number of users in the system be such that $K$ is a large multiple of $\lceil \frac{N}{M} \rceil 3g \log (N/M)$. Consider the case when users are divided into groups of size $K'= \lceil N/M\rceil 3g \log (N/M)$ and delivery scheme of Algorithm \ref{alg:ModifiedDelivery} is applied to each user group separately.  For any demand pattern, the expected total number of transmission required for all users is at most $ \frac{4}{3}\frac{K}{g+1} \left(1+o(1)\right)$. The file size needed is $F= O( {K' \choose g} (\log ({K' \choose g}))^2 \lceil N/M\rceil)  \approx O( \left(\frac{N}{M}\right)^{g+1} \left(3e \right)^g (\log (N/M))^{g+2} g^2 )$.
       \end{theorem}
       
    \textbf{Note:} The constant $e$ in the above requirement for file size comes due to bounding ${n \choose k}$ by $\left(\frac{ne}{k} \right)^k$. Other constants in the derivation can be relaxed if (\ref{eqn:highlevel}) can be strengthened which we do not do here. If $N/M = \Theta(K^{\delta})$ for some $0< \delta <1$ and $K$ large, then for a constant gain $g$, the above result requires $O \left(K^{\delta (g+1)} \right)$ packets whereas the previous best known uncoordinated random caching schemes require a file size of  $\Omega(\exp (K^{1-\delta}))$ for obtaining a gain of $2$.
   
  \section{Conclusion}
       We have analyzed random uncoordinated placement schemes along with clique cover based coded delivery schemes in the finite length regime for the caching-aided coded multicasting problem (or coded caching problem). This problem involves designing caches at user devices offline and optimizing broadcast transmissions when requests arise from a known library of popular files for worst case demand. The previous order optimal results on the number of broadcast transmissions for any demand pattern assumed that the number of packets per file is very large (tending to infinity). We showed that existing random placement and coded delivery schemes for achieving order optimal peak broadcast rate do not give any gain even when you have exponential number of packets. Further, we showed that to get a multiplicative gain of $g$ over the naive scheme of transmitting all packets, one needs $O((N/M)^g)$ packets per file for any clique cover based scheme where $N$ and $M$ are the library size and cache memory size respectively. We also provide an improved random placement and delivery scheme, that achieve this lower bound approximately.
       
     Future interesting research directions, to go beyond the bounds derived in this paper, may include designing improved deterministic caching schemes. This leads to several interesting research questions on designing very efficient coordinated deterministic placement schemes that go beyond the current ones and possible use of interference alignment inspired delivery schemes (instead of simple clique cover based delivery) that optimize the file size.
 
 \section*{Acknowledgement}
     We would like to thank MingFai Wong for many helpful discussions and in particular for his help with Theorem \ref{thm:concold}. 
    
\pagenumbering{arabic}
\bibliographystyle{IEEEtran}
\bibliography{Effbib}

\begin{thebibliography}{10}
\providecommand{\url}[1]{#1}
\csname url@samestyle\endcsname
\providecommand{\newblock}{\relax}
\providecommand{\bibinfo}[2]{#2}
\providecommand{\BIBentrySTDinterwordspacing}{\spaceskip=0pt\relax}
\providecommand{\BIBentryALTinterwordstretchfactor}{4}
\providecommand{\BIBentryALTinterwordspacing}{\spaceskip=\fontdimen2\font plus
\BIBentryALTinterwordstretchfactor\fontdimen3\font minus
  \fontdimen4\font\relax}
\providecommand{\BIBforeignlanguage}[2]{{%
\expandafter\ifx\csname l@#1\endcsname\relax
\typeout{** WARNING: IEEEtran.bst: No hyphenation pattern has been}%
\typeout{** loaded for the language `#1'. Using the pattern for}%
\typeout{** the default language instead.}%
\else
\language=\csname l@#1\endcsname
\fi
#2}}
\providecommand{\BIBdecl}{\relax}
\BIBdecl

\bibitem{Shanmugamallerton}
K.~Shanmugam, M.~Ji, A.~Tulino, J.~Llorca, and A.~Dimakis, ``Finite length
  analysis of caching-aided coded multicasting,'' in \emph{Communication,
  Control, and Computing (Allerton), 2014 52nd Annual Allerton Conference on},
  Sept 2014, pp. 914--920.

\bibitem{index2011global}
C.~V.~N. Index, ``Global mobile data traffic forecast update, 2010-2015,''
  \emph{White Paper, February}, 2011.

\bibitem{femto1}
K.~Shanmugam, N.~Golrezaei, A.~Dimakis, A.~Molisch, and G.~Caire,
  ``Femtocaching: Wireless content delivery through distributed caching
  helpers,'' \emph{Information Theory, IEEE Transactions on}, vol.~59, no.~12,
  pp. 8402--8413, Dec 2013.

\bibitem{molisch2014caching}
A.~F. Molisch, G.~Caire, D.~Ott, J.~R. Foerster, D.~Bethanabhotla, and M.~Ji,
  ``Caching eliminates the wireless bottleneck in video-aware wireless
  networks,'' \emph{arXiv preprint arXiv:1405.5864}, 2014.

\bibitem{golrezaei2012base}
N.~Golrezaei, A.~F. Molisch, and A.~G. Dimakis, ``Base-station assisted
  device-to-device communications for high-throughput wireless video
  networks,'' in \emph{Communications (ICC), 2012 IEEE International Conference
  on}.\hskip 1em plus 0.5em minus 0.4em\relax IEEE, 2012, pp. 7077--7081.

\bibitem{ji2013wireless}
M.~Ji, G.~Caire, and A.~F. Molisch, ``Wireless device-to-device caching
  networks: Basic principles and system performance,'' \emph{arXiv preprint
  arXiv:1305.5216}, 2013.

\bibitem{ji2013optimal}
------, ``Optimal throughput-outage trade-off in wireless one-hop caching
  networks,'' \emph{arXiv preprint arXiv:1302.2168}, 2013.

\bibitem{ji2014fundamental}
------, ``Fundamental limits of caching in wireless d2d networks,'' \emph{arXiv
  preprint arXiv:1405.5336}, 2014.

\bibitem{langberg2008hardness}
M.~Langberg and A.~Sprintson, ``On the hardness of approximating the network
  coding capacity,'' in \emph{Information Theory, 2008. ISIT 2008. IEEE
  International Symposium on}.\hskip 1em plus 0.5em minus 0.4em\relax IEEE,
  2008, pp. 315--319.

\bibitem{effros2012equivalence}
M.~Effros, S.~E. Rouayheb, and M.~Langberg, ``An equivalence between network
  coding and index coding,'' \emph{arXiv preprint arXiv:1211.6660}, 2012.

\bibitem{bar2011index}
Z.~Bar-Yossef, Y.~Birk, T.~Jayram, and T.~Kol, ``Index coding with side
  information,'' \emph{Information Theory, IEEE Transactions on}, vol.~57,
  no.~3, pp. 1479--1494, 2011.

\bibitem{maleki2012index}
H.~Maleki, V.~Cadambe, and S.~Jafar, ``Index coding: An interference alignment
  perspective,'' in \emph{Information Theory Proceedings (ISIT), 2012 IEEE
  International Symposium on}.\hskip 1em plus 0.5em minus 0.4em\relax IEEE,
  2012, pp. 2236--2240.

\bibitem{maddah2013fundamental}
M.~A. Maddah-Ali and U.~Niesen, ``Fundamental limits of caching,'' in
  \emph{Information Theory Proceedings (ISIT), 2013 IEEE International
  Symposium on}.\hskip 1em plus 0.5em minus 0.4em\relax IEEE, 2013, pp.
  1077--1081.

\bibitem{maddah2013decentralized}
------, ``Decentralized coded caching attains order-optimal memory-rate
  tradeoff,'' \emph{arXiv preprint arXiv:1301.5848}, 2013.

\bibitem{niesen2013coded}
U.~Niesen and M.~A. Maddah-Ali, ``Coded caching with nonuniform demands,''
  \emph{arXiv preprint arXiv:1308.0178}, 2013.

\bibitem{ji2014order}
M.~Ji, A.~M. Tulino, J.~Llorca, and G.~Caire, ``On the average performance of
  caching and coded multicasting with random demands,'' \emph{arXiv preprint
  arXiv:1402.4576}, 2014.

\bibitem{raab1998balls}
M.~Raab and A.~Steger, ``Òballs into binsÓÑa simple and tight analysis,'' in
  \emph{Randomization and Approximation Techniques in Computer Science}.\hskip
  1em plus 0.5em minus 0.4em\relax Springer, 1998, pp. 159--170.

\bibitem{jukna2001extremal}
S.~Jukna, \emph{Extremal combinatorics}.\hskip 1em plus 0.5em minus 0.4em\relax
  Springer, 2001, vol.~2.

\end{thebibliography}

\end{document}